\documentclass{llncs}
\usepackage[a4paper,hmargin=1in,tmargin=1.25in,bmargin=1.25in]{geometry}

\usepackage{amsmath}
\usepackage{amssymb}
\usepackage{url}
\usepackage{graphicx}
\usepackage{multirow}
\usepackage{subfigure}
\usepackage[ruled, vlined]{algorithm2e}
\usepackage{listings}
\usepackage{times}                               
\usepackage{verbatim}                            
\usepackage{graphicx}
\usepackage{tabularx, colortbl}
\usepackage{algo}
\usepackage{amsfonts}
\usepackage{units}
\usepackage{textcomp,xspace}
\usepackage{enumerate}
\usepackage{multirow}
\usepackage{xcolor}
\usepackage{authblk}
\definecolor{lightgray}{gray}{0.9}

\newcommand{\bigO}{\ensuremath{\mathcal{O}}}

\begin{document}

\title{Complexity Analysis of Reversible Logic Synthesis}


\author{Anupam Chattopadhyay\inst{1} \and Chander Chandak\inst{2} \and Kaushik Chakraborty\inst{3}}
\institute{MPSoC Architectures Research Group, RWTH Aachen University, Germany \and IIT Kharagpur, India \and Indian Statistical Institute, Kolkata, India}

\sloppy

\maketitle

\begin{abstract}
Reversible logic circuit is a necessary construction for achieving ultra low power dissipation as well as for prominent post-CMOS computing technologies such as Quantum computing. 
Consequently automatic synthesis of a Boolean function using elementary reversible logic gates has received significant research attention in recent times, creating the domain of 
reversible logic synthesis. In this paper, we study the complexity of reversible logic synthesis. The problem is separately studied for bounded-ancilla and ancilla-free optimal 
synthesis approaches. The computational complexity for both cases are linked to known/presumed hard problems. Finally, experiments are performed with a shortest-path based reversible 
logic synthesis approach and a (0-1) ILP-based formulation.
\end{abstract}

%

\section{Introduction} \label{sec:introduction}
\begin{sloppypar}
Asymptotically zero power dissipation can be achieved by performing computation in a reversible manner, which implies~\cite{Bennett} that 
logical reversibility must be supported to achieve physical reversibility. Consequently, major research attention is given towards the synthesis 
of a Boolean function using basic reversible logic gates, a problem otherwise known as reversible logic synthesis. 
\end{sloppypar}

\subsection{Background}
\begin{sloppypar}
An $n$-variable Boolean function $f$ is a mapping $f : GF(2^n) \rightarrow GF(2)$. An alternative representation of a Boolean function $f$ is a mapping 
$f : \{0, 1\}^n \rightarrow \{0, 1\}$, which is known as the truth table representation. Using any basis of $GF(2^n)$, we can express each $x \in GF(2^n)$ 
as an $n$-tuple $(x_1 x_2 \ldots x_n)$, $x_i \in GF(2)$, $i = 1, \ldots, n$. Thus we can derive the truth table representation from the former representation. 
\end{sloppypar}

\begin{sloppypar}
Alternatively, an $n$-variable Boolean function $f(x_1, \ldots, x_n)$ can be considered to be 
a multivariate polynomial over $GF(2)$. This polynomial can be expressed as a sum-of-products 
representation of all distinct $k$-th order products $(0 \leq k \leq n)$ of the variables. This 
representation of $f$ is called the {\em Algebraic Normal Form} (ANF) of $f$.
\end{sloppypar}

\begin{sloppypar}
An $n$-variable Boolean function is \textit{reversible} if all its output patterns map uniquely to an input pattern and vice-versa. It can be expressed as 
an $n$-input, $n$-output bijection or alternatively, as a permutation function over the truth value set $\{0, 1, \ldots 2^{n-1}\}$. 
\end{sloppypar}

\begin{sloppypar}
The problem of reversible logic synthesis is to map such a reversible Boolean function on a reversible logic gate library. Reversible gates are characterized by their 
implementation cost in Quantum technologies, which is dubbed as Quantum Cost (QC) \cite{maslov_benchmark,barenco}. Prominent classical reversible logic gates include 
NOT, Feynman (or CNOT), Toffoli (or CCNOT) gates. Controlled NOT gates can be generalized as $Tof_n$ gate, where first $n-1$ variables are used as control lines. There 
are other reversible logic gate families such as Fredkin gate family ($Fred_n$), where first $n-2$ variables are used as control lines and last $2$ lines undergo swap, 
when the logical conjunction of first $n-2$ variables is true. 
\end{sloppypar}

\begin{itemize}
\item NOT gate: 
\begin{equation}
f(A) = \overline{A}
\end{equation}
\item Controlled NOT gate: Also known as Feynman or CNOT gate. 
\begin{equation}
f(A) = A, f(B) = A \oplus B
\end{equation}
\item Controlled Controlled NOT gate: Also known as Toffoli gate. 
\begin{equation}
f(A) = A, f(B) = B, f(C) = A\cdot B \oplus C
\end{equation}
This gate can be generalized with $Tof_n$ gate, where first $n-1$ variables are used as control lines. NOT gate is denoted as $Tof_0$ gate.
\item Swap gate: 
\begin{equation}
f(A) = A, f(B) = A
\end{equation}
This gate can be generalized to the Fredkin gate family ($Fred_n$), where first $n-2$ variables are used as control lines and last $2$ lines undergo swap.
\begin{equation}
f(A) = A, f(B) = \overline{A}\cdot B + A\cdot C, f(C) = \overline{A}\cdot C + A\cdot B
\end{equation}
\end{itemize}

\begin{sloppypar}
Multiple sets of reversible gates form an universal gate library for realizing classical Boolean functions such as, (i) NCT: NOT, CNOT, Toffoli. 
(ii) NCTSF: NOT, CNOT, Toffoli, SWAP, Fredkin. (iii) GT: $Tof_n$. (iv) GTGF: $Tof_n$ and $Fred_n$. Generally, gates with more control lines incur higher QC.
\end{sloppypar}

\begin{figure}[hbt]
 \begin{center}
   \includegraphics[width=90mm]{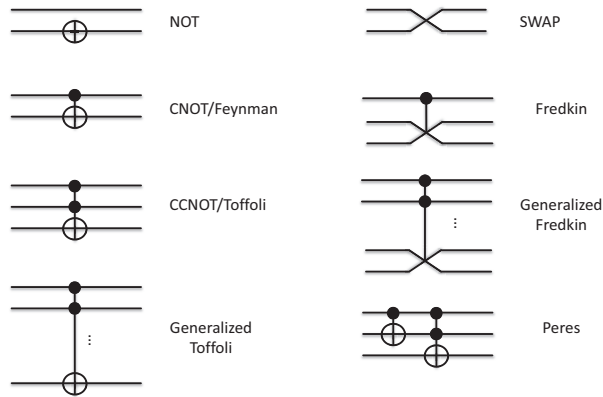}
   \caption{Reversible Logic Gates}
   \label{fig:basic_gates}
 \end{center}
\end{figure}

\begin{sloppypar}
Reversible logic synthesis begins from a given $n$-variable Boolean function, which can be irreversible. The first step is to convert it to a reversible Boolean function 
by adding additional constant input bits (known as \textit{ancilla inputs}). If these constant input bits are not restored to their original values at the end of computation, 
then these are referred as \textit{ancilla}.
\end{sloppypar}

\section{Related Work and Motivation} \label{sec:related}
\begin{sloppypar}
Reversible logic synthesis methods can be grossly classified as following. First, \textit{optimal} implementation is found by making a step-by-step exhaustive enumeration 
or by formulating the reversible logic synthesis as a SAT problem~\cite{wille_sat}. Optimal implementation up to $4$-variable Boolean functions are 
known~\cite{shende_3var,golubitsky_4var}. Pseudo-optimal synthesis for linear Boolean functions is proposed at~\cite{rev_linear}. Second, \textit{transformation-based 
methods}~\cite{mmd,chander_la_mmd}, which apply controlled transformations to map output Boolean functions to input Boolean functions. Third category includes methods based 
on {\em decision diagrams}~\cite{wille_bdd}, each node of the decision diagram is converted to an equivalent reversible circuit structure. While these methods are known to 
scale for large Boolean functions, they introduce too many garbage bits. Finally, {\em search-based methods} start from an Exclusive Sum-of-Product (ESOP) representation of 
the Boolean function. The ideal candidates for reversible logic realization at each depth is chosen via heuristic search~\cite{esop_gupta}. Another synthesis, based on ESOP 
formulation, is proposed in~\cite{esop_cube_nayeem}. There a heuristic ESOP minimization tool~\cite{mischenko_xorcism} is used for generating reversible circuits with fixed 
number of ancilla lines. For a detailed survey of reversible logic synthesis, reader is kindly referred to~\cite{wille:rev_book,markov_survey}.
\end{sloppypar}

\begin{sloppypar}
\textbf{Motivation:} The promise held by Quantum computing is driving the research of reversible logic synthesis. Despite several major advancements of this research, the 
complexity study of the problem is not done so far. This is important since, optimal, comprehensive reversible circuit synthesis is done up to only all $4$-variable 
Boolean functions~\cite{golubitsky_4var} and up to selected $6$-variable Boolean functions~\cite{wille_sat,wille_smt}.
\end{sloppypar}

\begin{sloppypar}
The rest of this paper is organized as following. In the following section~\ref{sec:obj}, the performance objectives for reversible logic synthesis are discussed. 
In section~\ref{sec:ancilla_bounds}, analayis of computational complexity for bounded-ancilla reversible logic synthesis is presented. 
This is followed by the computational complexity analysis for ancilla-free reversible logic synthesis in section~\ref{sec:ancilla_free_bounds}. 
Two different formulations for exactly solving the reversible logic synthesis problem is presented in sections~\ref{sec:shortest_path} and~\ref{sec:ilp_formulation} 
respectively. Experimental results with these methods are reported in section~\ref{sec:results}. The paper is concluded and future works are outlined in section~\ref{sec:conclusion}.
\end{sloppypar}

\section{Performance Objectives} \label{sec:obj}
\begin{sloppypar}
Before presenting the complexity analysis of reversible logic synthesis, it is necessary to develop a better understanding of the performance objectives. 
The major performance objectives of reversible logic synthesis are gate count, logical depth, Quantum cost, ancilla and garbage count. 
\end{sloppypar}

\subsubsection{Ancilla Count}
\begin{sloppypar}
In literature, either ancilla count or garbage count or total line count is reported. Ancilla-free synthesis is achievable by several methods~\cite{mmd,wille_sat}, whereas 
several synthesis methods compromise ancilla count for achieving scalability~\cite{wille_bdd} or lowering QC~\cite{esop_cube_nayeem}. 
\end{sloppypar}

\subsubsection{Gate Count}
\begin{sloppypar}
Gate count is used as a major performance objective in many synthesis flows~\cite{mmd,esop_cube_nayeem} due to its simple cost annotation. However, several 
composite gates (e.g. Peres gates) are presented both as a single gate or multiple gates in literature. For example, Peres gates are considered as a single 
gate for multiple benchmark circuits presented in~\cite{maslov_benchmark}, whereas the gate counts reported for the Quantum arithmetic circuits in~\cite{takahashi_quant_add} 
count Peres gates as a collection of Toffoli and CNOT gates. To avoid confusion, library-specific gate counts are reported in~\cite{esop_gupta}.
\end{sloppypar}

\subsubsection{Quantum Cost}
\begin{sloppypar}
Quantum Cost (QC) for a reversible gate is technology-dependent. A set of primitive gates are actually implemented on experimental quantum circuits, which are used to build more 
complex gates and thus the QC is computed by simply adding those primitive gates. Early studies of primitive reversible gates are presented at~\cite{barenco}, which is adopted 
for the derivation of QC for generalized Toffoli and Fredkin gates in~\cite{maslov_benchmark}. Recent advances at experimental quantum computing~\cite{universal_qc_surface} shows 
that the QC of negative-control gates are comparable to that of positive-control gates, thereby spawning several works on reversible logic synthesis with mixed-polarity gate 
libraries~\cite{kerntopf_4var_mixed_pol}. An improved circuit, in terms of less Quantum cost, for generalized Toffoli gates is presented in~\cite{maslov_improved_qc}. In a 
significant new result, improved QC for generalized Peres gates as well for generalized Toffoli gates are presented in~\cite{kerntopf_gen_peres}. Reduced QC values are 
reported by introducing a new gate library in~\cite{wille_ncvv1}.
\end{sloppypar}

\begin{sloppypar}
Though there are established performance metrics for evaluating the quality of a reversible circuit, yet the continuously evolving Quantum technology needs to be considered 
for a fair benchmarking. The synthesis tools should also be adaptable to new performance objectives and their trade-offs.
\end{sloppypar}

\section{Complexity Analysis of Bounded-Ancilla, Exact Reversible Logic Synthesis} \label{sec:ancilla_bounds}
\begin{sloppypar}
The complexity of exact Sum-of-Product (SOP) minimization, in terms of number of literals, where the input is an incomplete or complete truth-table specification is studied 
in~\cite{masek_np_complete} and~\cite{umans_complexity}. It is shown that the problem - {\em Is there a SOP representation of a given Truth-table specification with at most 
k literals?} - is NP-complete. The proof is based on the fact that the well-known exact procedure of SOP minimization, Quine-McCluskey includes a function call to the minimum 
cover problem. Minimum cover problem is shown to be NP-complete in 1972~\cite{karp_np_complete}. We introduce few definitions, which are needed for rest of this section.
\end{sloppypar}

\begin{definition}
Given a set of elements $\mathcal{U} = \{1,2,\cdots,m\}$ and a set $\mathcal{S}$ of $n$ sets, whose union equals $\mathcal{U}$, the minimum set covering problem is to identify 
the smallest subset of $\mathcal{S}$, the union of which contains all elements of $\mathcal{U}$.
\end{definition}

\begin{definition}
Exact cover or exact set cover problem is a decision problem, where the goal is to find a subset of $\mathcal{S}$, the union of which contains all elements and each element 
of $\mathcal{U}$ is covered exactly by one subset. Exact cover problem returns false if no such set exists.
\end{definition}

\begin{sloppypar}
For example, given a set $\mathcal{U} = \{1,2,3,4,5\}$ and $\mathcal{S} = \{\{1,2\}, \{2,3,4,5\}, \{3,4\}, \{5\}\}$, $\{\{1,2\}, \{2,3,4,5\}\}$ returns the minimum set cover 
and $\{\{1,2\}, \{3,4\}, \{5\}\}$ returns an exact cover. The exact cover problem and the decision version of 
set covering problem are NP-complete~\cite{karp_np_complete}.
\end{sloppypar}

\begin{sloppypar}
The worst-case complexity of SOP formulation ($2^{n-1}$) is more than the corresponding complexity of the ESOP formulation $3\cdot2^{n-3}$~\cite{sasao_esop_bound}, for an 
$n$-variable Boolean function. This led researchers to look for efficient ESOP minimization flows and several exact ESOP minimization algorithms have been 
presented~\cite{perkowski_helliwell,steinbach_snf}. For the computational complexity of reversible logic synthesis, 
exact ESOP minimization holds a clue, since for an $n$-variable Boolean functions, an ESOP formulation with $k$ cubes directly corresponds to a reversible circuit 
realization with $k$ $Tof_n$ gates and at most $n$ ancilla lines, each corresponding to one output function. It might be argued that, the function of an ESOP is independent 
of the order of its product terms whereas, the function of a reversible circuit is affected by the order of the sequence of $Tof_n$ gates. However, that does not hold true 
for a bounded-ancilla reversible logic circuit, where the ancillae ensure that the inputs remain unchanged. The NOT gates inserted between the $Tof_n$ gates can be ignored if 
mixed-polarity $Tof_n$ gates are considered.
\end{sloppypar}

This is exemplarily shown in the reversible circuit realizations depicted in Fig.~\ref{fig:esop_tof} for the Boolean 
functions $f(a, b, c) = a \oplus a\cdot b \oplus b\cdot c$ and $f(a, b, c) = a \oplus c \oplus a\cdot b$ respectively.

\begin{figure}[hbt]
 \begin{center}
   \includegraphics[width=100mm]{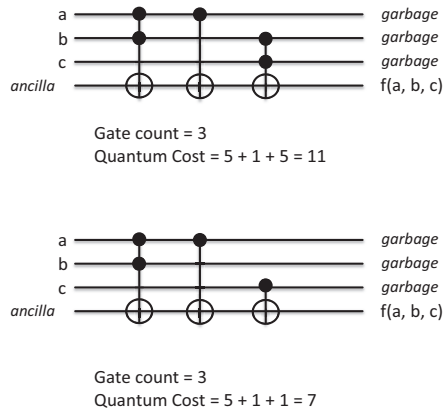}
   \caption{ESOP-based Reversible Circuit Realization}
   \label{fig:esop_tof}
 \end{center}
\end{figure}

\begin{sloppypar}
Exact minimization of ESOP is well-studied in literature~\cite{sasao_exact,steinbach_snf}. In~\cite{sasao_exact}, exact ESOP minimization is done based on the shortest path 
enumeration of a decision diagram, which leads to the complexity of $\bigO({2^{3^n}})$ for an $n$-variable Boolean function. Computational complexity analysis is not done there. 
Another exact ESOP minimization is proposed by Steinbach and Mishchenko~\cite{steinbach_snf}. There, the exact ESOP formulation is reduced to a problem of a so called 
{\em coverage matrix}. We adopt the exact minimization procedure from~\cite{steinbach_snf} for studying the computational complexity. In the following, we briefly describe 
the formulation. 
\end{sloppypar}

\begin{sloppypar}
First, the SOP or truth table formulation is converted to an ESOP formulation. For a SOP to ESOP formulation, the following transformation can be used.
\begin{equation}
a + b = a \oplus b \oplus a\cdot b 
\end{equation}
It is straightforward to show that for a SOP formulation with $k+1$ cubes, an ESOP formulation with $2^{k+1}-1$ literals is obtained. Alternatively, 
constructing an ANF from a Boolean truth table specification can be done using $O(n2^n)$ operations with standard algorithm. 
\end{sloppypar}

\begin{sloppypar}
The input to the algorithm is the truth table $f = [f(0) f(1) f(2) \ldots f(2^n -1)]$, and the output is the coefficient vector of the 
canonical Algebraic Normal Form (ANF), represented as $C = [c_0 c_1 c_2 \ldots c_{2^n -1}]$. Only if $c_j$ = 1, where 0 $\leq$ j $\leq$ 
$2^n - 1$, then the monomial $x_0^{j_0}x_1^{j_1}\cdots x_{n-1}^{j_{n-1}}$ exists in the ANF of $f$, where ($j_0,j_1\cdots j_{n-1}$) is 
the binary representation of index $j$. For an $n$-variable Boolean function, $C = fA_n$, where $A_n$ can be computed as following.
$
A_n = \begin{bmatrix}
       A_{n-1} & A_{n-1} \\
        0       & A_{n-1}
       \end{bmatrix}
, where A_0 = 1.
$
\end{sloppypar}

\begin{figure}[hbt]
  \begin{center}
    \includegraphics[width=80mm]{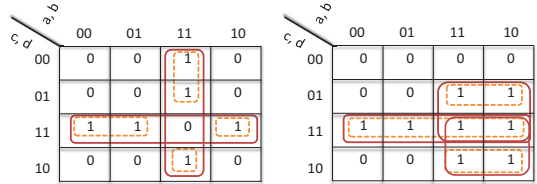}
  \end{center}
  \caption{ESOP Minimization via Karnaugh Map}
  \label{fig:off_cover}
\end{figure}

\begin{sloppypar}
First, the SOP or truth table formulation is converted to an ESOP formulation. Based on the ESOP formulation, a Special Normal Form (SNF) is introduced in~\cite{steinbach_snf}. 
It was shown that it is also a canonical ESOP representation like ANF. Whereas ANF is a positive-polarity expression, SNF is mixed-polarity. For exact ESOP minimization, 
a coverage matrix is constructed, where the rows and columns correspond to all possible ($3^n$) cubes of an $n$-variable Boolean function $f$ and the cubes of SNF($f$) 
respectively. The goal of exact ESOP minimization is to include the fewest row-cubes from $f$ to cover all the column-cubes from SNF($f$), which reduces to a variation 
of the minimum set covering problem. However, there is a subtle difference between the minimum set covering problems encountered in ESOP and SOP minimization. For the SOP 
formulation, the goal is to cover the ON-set of output function with minimum number of cubes, where multiple cubes can cover a ON-value as in {\em inclusive OR}. For ESOP 
formulation, an ON-value can be covered by odd number of cubes and an OFF-value can be covered by even number of cubes as in {\em exclusive OR}. This is exemplarily shown in 
the Karnaugh-map diagrams in Fig.~\ref{fig:off_cover}. The dotted lines shows the exact cover without multiple covering of an element. The solid lines shows the covers with 
multiple covering. Note that while multiple cover of the OFF-set in the left-hand figure decreases the number of cubes, whereas multiple covers of the ON-set elements leaves 
the number of cubes intact but, increases the sizes of the cubes, thus reducing literals. Based on this, we propose the following formulation of ESOP minimization.
\end{sloppypar}

\begin{sloppypar}
ESOP minimization is a minimum set covering problem, where the union is defined as following. \\
$Z = X \cup Y$, where \\
$e \in Z$ if $e \in X \wedge e \notin Y$ \\
$e \in Z$ if $e \in Y \wedge e \notin X$ \\
We refer to the problem as XOR-SET-COVER.
Formally, the XOR-SET-COVER problem is:
\begin{description}
 \item \textbf{INPUT:} $\langle \mathcal{U}, \mathcal{A}, \mathcal{K} \rangle$ where $\mathcal{A}$ is a set of subsets, consisting of members from $\mathcal{U}$ and $\mathcal{K} \in \mathbb{N}$
 \item \textbf{QUESTION:} Does $\mathcal{U}$ have a $\mathcal{A}$-cover of size $\mathcal{K}$, where the union of two sets is defined as above ?
\end{description}
\end{sloppypar}

\begin{claim}
XOR-SET-COVER is in NP. 
\end{claim}
\begin{proof}
Let us say, $\mathcal{B}$ is a proposed $\mathcal{A}$-cover solution. We propose the following verification flow.
\begin{itemize}
 \item $\mathcal{B}$ is subset of $\mathcal{A}$
 \item $\left\vert \mathcal{B} \right\vert$ $\leq$ $\mathcal{K}$
 \item $\forall u \in \mathcal{U} \exists B \in \mathcal{B} \mid u \in B \wedge count(u)$ is odd.
\end{itemize}
The verifier runs in time polynomial in the length of the $\mathcal{B}$.
\end{proof}

\begin{theorem}
XOR-SET-COVER problem is NP-Complete.
\end{theorem}

\begin{proof}
Suppose we have a black box $M_{XS}$ which takes a three tuple $\langle U, A, K \rangle$ as input and returns $1$ \textit{iff} $U$ has a XOR SET $A$ COVER of size $K$ otherwise 
it returns $0$. Here, $U,A,K$ is same as defined previously. Now using $M_{XS}$ as a subroutine we have to solve the EXACT SET COVER problem. So, the input for EXACT SET COVER 
problem is also a three tuple $\langle U,A,K \rangle$ and we have to answer whether $U$ has a set $A$-cover of size $K$. If for any input $\langle U,A,K \rangle$, the $M_{XS}$ gives 
$1$ as output, then we can conclude that $\langle U,A,K \rangle$ is an exact set cover. Now the problem is when $M_{XS}$ gives $0$ as output. Then there are two following possibilities.

\begin{enumerate}
\item $U$ has exact $A$-set cover of size $K$, let $S_K \subseteq A$ be a solution i.e, union of the elements of $S_K$ gives the set $U$. As $M_{XS}$ is giving the 
output $0$, so, we can say that if we perform the $XOR$ operation defined previously among the elements of $S_K$, then there exists at least one element $u \in U$, which is present 
in even number of elements of $S_K$.
\item $U$ doesn't have exact $A$-cover of size $K$. 
\end{enumerate} 

Among the above two possibilities, for the first possibility, we can take an element $u \in U$ and remove it from one element of $A$, say $A_u ^i$ be the new set of subsets 
and check whether $U$ has xor set $A_u ^i$-cover of size $K$. Then, we can repeat this process for each element of $A$ for a fixed $u$ and then repeat this process for each 
element $u \in U$. If for every input $M_{XS}$ gives $0$ as output then conclude that $U$ doesn't have exact $A$-set cover of size $K$ otherwise, if for at least one single query 
to $M_{XS}$ it give the output $1$ conclude that $U$ has exact $A$-set cover of size $K$. So, here we have reduced the EXACT SET COVER problem to XOR-SET-COVER problem. As, 
EXACT SET COVER is known NP-complete problem and XOR-SET-COVER is in NP, so XOR-SET-COVER  problem is NP-complete.
\end{proof}

\begin{sloppypar}
It is important to show how the XOR-SET-COVER is applicable to the problem at hand. We start by enumerating all possible cubes w.r.t. $n$ input variables of the Boolean function 
$f$. We arrange a coverage matrix, where all possible cubes ($3^n$) are arranged row-wise and the cubes from the canonical SNF or ANF are arranged column-wise. Let us assume 
that we have $\mathcal{I}$ and $\mathcal{J}$ elements in the ON-set and in the OFF-set respectively. Further, the cubes belonging to the OFF-set are also added to the columns. 
Hence, we get a $3^n \times 3^n$ coverage matrix. The elements of the matrix are assigned a value of $0$ if a row-cube does not cover a column-cube. Otherwise, the element 
value is assigned to be $1$ if the column-cube belongs to the ON-set and a value of $-1$ if it belongs to the off-set. An exemplary coverage matrix is shown in 
Table~\ref{tab:coverage_matrix}, where the ON-set column headers are marked bold. 
\end{sloppypar}

\begin{sloppypar}
A valid solution for the coverage matrix is a XOR-SET-COVER, where the columns belonging to ON-set can be replicated odd number of times and the columns belonging to the 
OFF-set can be replicated even number of times. Note that, there is a upper limit of replication~\cite{steinbach_snf}.The ON-set cubes can be replicated up to $2^n-1$ 
times and OFF-set cubes can be replicated up to $2^n$ times.
\end{sloppypar}

\begin{table}[h]
    \begin{center}
        \caption{Exemplary $2$-variable Coverage Matrix~\cite{steinbach_snf}}
        \label{tab:coverage_matrix}
        \begin{tabular}{c|c|c|c|c|c|c|c|c|c} \hline 
    & - -   & \textbf{-0}   & -1 & \textbf{0-} & 00 & \textbf{01} & \textbf{1-} & \textbf{10} & \textbf{11} \\ \hline
- - & 0     & 0  & 0     & 0  & $-$1  & 1  & 0  & 1  & 1  \\
-0  & 0     & 0  & 0     & 1  & 0     & 1  & 1  & 0  & 1  \\
-1  & 0     & 0  & 0     & 1  & $-$1  & 0  & 1  & 1  & 0  \\
0-  & 0     & 1  & $-$1  & 0  & 0     & 0  & 0  & 1  & 1  \\
00  & $-$1  & 0  & $-$1  & 0  & 0     & 0  & 1  & 0  & 1  \\
01  & $-$1  & 1  & 0     & 0  & 0     & 0  & 1  & 1  & 0  \\
1-  & 0     & 1  & $-$1  & 0  & $-$1  & 1  & 0  & 0  & 0  \\
10  & $-$1  & 0  & $-$1  & 1  & 0     & 1  & 0  & 0  & 0  \\
11  & $-$1  & 1  &  0    & 1  & $-$1  & 0  & 0  & 0  & 0  \\
\hline 
        \end{tabular}
    \end{center}
\end{table}

\begin{sloppypar} 
Fig.~\ref{fig:one_cover} shows an example for $4$ or $6$ variable Boolean function, where the element $1$ is covered $3$ times resulting in a reduced cube count.
\end{sloppypar}

\begin{figure}[hbt]
 \begin{center}
   \includegraphics[width=80mm]{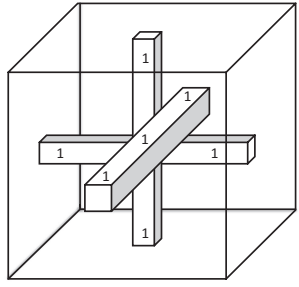}
   \caption{Multiple covers of ON-set}
   \label{fig:one_cover}
 \end{center}
\end{figure}

\begin{lemma}
To decide if bounded-ancilla reversible logic synthesis is achievable with a specified gate count is $NP-complete$. 
\end{lemma}
\begin{proof}
This follows the previous theorem directly.
\end{proof}

\begin{lemma}
To decide if bounded-ancilla reversible logic synthesis is achievable with a specified Quantum Cost is $NP-complete$. 
\end{lemma}
\begin{proof}
The decision version of bounded-ancilla reversible logic synthesis with an user-specific QC is an instance of the decision version of ESOP minimization, where the element 
values ($1$, $-1$ and $0$) are replaced with the number of literals ($l$, $-l$ and $0$) in the particular row-cube. The number of literals directly indicate the $Tof_n$ to be 
employed and hence, reflect the QC of the circuit. The decision version of ESOP minimization i.e. XOR-SET-COVER is shown to be $NP-complete$. Hence, the proof.
\end{proof}


\section{Complexity Analysis of Ancilla-free, Exact Reversible Logic Synthesis} \label{sec:ancilla_free_bounds}
\begin{sloppypar}
Exact synthesis for a reversible Boolean specification is so far proposed by modeling the synthesis as a satisfiability (SAT) problem~\cite{wille_sat} as well as modeling 
the problem in terms of symbolic reachability analysis~\cite{hung_reachability}. However, none of these works attempt the complexity analysis of reversible logic synthesis. 
\end{sloppypar}

\begin{sloppypar}
We can easily show that the ancilla-free reversible logic synthesis is in the class of NP.
\end{sloppypar}

\begin{claim}
The decision problem of whether an ancilla-free reversible circuit with $\mathcal{K}$ gates exist, for a given permutation $\pi$ is verifiable in polynomial time. Hence, the problem is in NP.
\end{claim}
\begin{proof}
Let us say, $\mathcal{C}$ is a proposed $\mathcal{K}$-gate circuit. We propose the following verification flow.
\begin{itemize}
 \item $\mathcal{C}$ results in a permutation $\sigma$
 \item $\sigma^{-1}\cdot\pi = I$
\end{itemize}
The verifier runs in time polynomial in the cardinality of $\pi$.
\end{proof}

\begin{sloppypar}
In order to explore the complexity of ancilla-free reversible logic synthesis, we connect the reversible logic synthesis problem to sorting.
A reversible Boolean function can be defined as an ordered set of integers corresponding to a permutation $\pi$ of its domain. The reversible circuit, 
when traversed from output towards input, essentially converts the permutation $\pi$ to an Identity $I$. We explain this with an example.
\end{sloppypar}

\begin{sloppypar}
\begin{definition}
Let $S$ be an arbitrary nonempty set. A bijection of $S$ onto itself is called a permutation of $S$.
\end{definition}

Following Cauchy's two-line notation where all the mapping $x:S \rightarrow S$ can be written as 

\begin{equation}
      \left(
        \begin{array}{lcl}
        \alpha_1 \alpha_2 \alpha_3 ... \alpha_n\\
        \beta_1  \beta_2  \beta_3  ... \beta_n
        \end{array}
      \right) \\
\end{equation}

where the top row is some enumeration of the points of $S$ and $β_i$ is the image of $\alpha_i$ under $x$ for each $i$. This can be 
alternatively formalized as a functional digraph.

\begin{definition}
Given a function $f:[0,n] \rightarrow [0,n]$, the functional digraph $G(f) = (V,E)$ associated with f is a directed graph with 
$V= \left\{0,...,n\right\}$ and $E = \left\{(v,f(v)) \text{for each }v \in V\right\}$   
\end{definition}

The functional digraph of the reversible function $\left\{7,1,4,3,0,2,6,5\right\}$ is given in figure \ref{fig:func_digraph}.

\begin{figure}[hbt]
 \begin{center}
   \includegraphics[width=80mm]{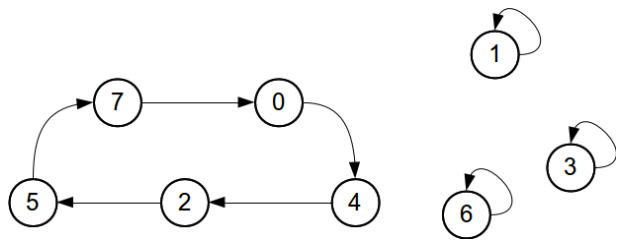}
   \caption{Functional digraph}
   \label{fig:func_digraph}
 \end{center}
\end{figure}

For a permutation $\pi$ of ${0, 1, \cdots , n}$, $G$ is a collection of disjoint cycles, since the in-degree and out-degree
of each vertex is exactly one. Each cycle can be written as a $k$−tuple $c = (a_1,a_2, \cdots, a_k)$, where $k$ is the length of
the cycle and $a_{i+1} =  c(a_i)$.

In other words, the permutation cyclically shifts all entries in $\left\{a_1,a_2, \cdots, a_k\right\}$  and keeps all other elements fixed.

$$a_1\rightarrow a_2 \rightarrow a_3 \rightarrow \cdots a_k \rightarrow a_1$$

{\em Transposition} is defined as a cycle of length $2$, when it is applied on two adjacent blocks. Otherwise, it is defined as {\em block interchange}. 
Given a permutation $\pi$, a block interchange with parameters $\{i,j,k,l\}$, where $1 \leq i < j \leq k < l \leq n+1$, is applied to $\pi$ by 
exchanging the blocks $[\pi_i,\cdots,\pi_{j-1}]$ and $[\pi_k,\cdots,\pi_{l-1}]$. A special case of this is transposition, where $j=k$.
A sorting $s$ of a permutation $\pi$ is a sequence of transpositions that transform $\pi$ into $I$, where $I$ denotes the identity element of 
$\mathbb{S}_n$, i.e., $s\cdot\pi = I$. The length of the sequence of transpositions to perform sorting is known as {\em transposition distance}, 
whereas the length of the sequence of block interchanges is called {\em minimum block interchange distance}.
\end{sloppypar}

\begin{sloppypar}
The fact that a permutation can be decomposed into transpositions was first applied into genome sequence comparison~\cite{bafna_sorting}, where 
the transposition distance ($d_t(\pi)$), i.e., number of transpositions needed to reach from one permutation to $I$ is studied. The problem of 
sorting by transposition is defined as following - {\em given a permutation $\pi$ and an integer $k$, is $d_t(\pi) \leq k$} ? We briefly state two 
results connecting reversible logic gates and sorting.
\end{sloppypar}

\begin{lemma}
When the reversible Boolean function is represented as a permutation $\pi$, the application of a $Tof_n$ gate is equivalent to a block interchange.
\end{lemma}
\begin{proof}
Any uni-polarity/multi-polarity $Tof_n$ gate introduces change in exactly one bit position, which is nothing but exchange of two elements in $S$. 
This is equivalent to a $2$-cycle, though not necessarily for the adjacent elements. Hence, a $Tof_n$ gate introduces a block interchange.
\end{proof}

\begin{lemma}
A block interchange defined by the parameters $\{i,j,k,l\}$, where $j = i+1, j<k, l=k+1$ incurs a definite non-negative cost in terms of $Tof_n$ gate count.
\end{lemma}
\begin{proof}
The block interchange swaps two elements $\pi_i$ and $\pi_k$. Let the number of bits required to represent an element of the permutation be $n$. A multi-polarity 
$Tof_n$ gate will be active to a permutation element iff all positive and negative control lines evaluate to one and zero correspondingly. In that case, the 
target bit is inverted. Hence a $Tof_n$ gate can perform a block interchange of two permutation elements with a Hamming distance of $1$. Let the Hamming distance 
between $\pi_i$ and $\pi_k$ be $H(\pi_i, \pi_k) = h_{i,k}$. 
To achieve the swap, we first transform $\pi_i$ to $\pi_k$ by applying a series of $Tof_n$ gates. We identify the bits of $\pi_i$ that differ from corresponding bits 
of $\pi_k$. For each bit, we apply a $Tof_n$ gate whose target line maps to that bit and the control lines corresponds to other bits of $\pi_i$. After each step, the 
Hamming distance will decrease by $1$. After $h_{i,k}$-th step, the combined circuit will transform $\pi_i$ to $\pi_k$. The $h_{i,k}$-th $Tof_n$ gate will also invert 
a bit of $\pi_k$ as the Hamming distance between $\pi_k$ and intermediate output of $\pi_i$. If we apply these series of $Tof_n$ gates in reverse order it will transform 
$\pi_k$ to $\pi_i$. As there is one common gate, total number of gates needed are $2h_{i,k}-1$.

As we are applying each $Tof_n$ gate twice, once in direct and once in reverse order, this will cancel out the introduced changes in the permutation elements other 
than $\pi_i$ and $\pi_k$.
\end{proof}

\begin{sloppypar}
Since the introduction of sorting by transpositions by Bafna and Pevzner~\cite{bafna_sorting}, it has been studied with the context of varying applications, e.g., genome 
sequencing and rank permutation coding. The special case of sorting by block interchange is shown to be solvable in polynomial time~\cite{christie_sorting_block}. However, 
this is not applicable to reversible logic synthesis since, even though each $Tof_n$ gate introduces a block interchange, the block interchange required by the polynomial-time 
algorithm may not be realizable with a single $Tof_n$ gate. Furthermore, the algorithm in~\cite{christie_sorting_block} does not include the cost of a block interchange.
For reversible logic synthesis, the cost for different block interchanges could be different.
\end{sloppypar}

\begin{sloppypar}
The problem of sorting by transposition is shown to be NP-hard~\cite{sorting_difficult}. This is, again, not directly applicable to the case at hand since, sorting by transposition 
attempts to find {\em minimum transposition distance} as well as imposes a restriction that the sorting needs to be achieved by transpositions only. Whereas, for the reversible 
logic synthesis, the output Boolean function, could be sorted by a series of block interchanges. With this context, we define the following problem, which is mappable to reversible 
logic synthesis directly.
\end{sloppypar}

\begin{sloppypar}
Ancilla-free reversible logic synthesis is {\em sorting by minimum-cost block interchange problem}, where the cost of exchanging two single-element blocks and the block cardinality 
are as following. \\
$B_i = \{\pi_i\}, B_j = \{\pi_j\}$, \\
$Cost(B_i, B_j) = 2h_{i,j}-1$, and \\
$\left\vert{B}\right\vert = 1$
\end{sloppypar}

\begin{sloppypar}
A related problem of {\em sorting by cost-constrained transpositions}~\cite{sorting_cost_constrained_tp} is studied recently, for which the complexity is yet to be established. 
\end{sloppypar}

\section{Shortest Path Problem Formulation} \label{sec:shortest_path}
\begin{sloppypar}
The most prominent heuristics for ancilla-free reversible logic synthesis adopt a branch-and-bound heuristic~\cite{esop_gupta} or apply repeated transformations~\cite{mmd,chander_la_mmd}. 
The former approach is built on top of an ESOP synthesis, hence reduces to the previously identified complexity. The latter approaches do not link their result to optimality and hence we 
refrain from the corresponding complexity analysis. In the following, we present reversible logic synthesis as a shortest path problem. This has been briefly explored 
in ~\cite{markov_modular_mult} without the complexity analysis. 
\end{sloppypar}

\begin{sloppypar}
It is shown in~\cite{wille_sat} that for an $n$-variable Boolean function, there exist $n\cdot2^{n-1}$ Toffoli gates. Along similar lines, the total number of Fredkin gates 
can be shown to be $\binom{n}{2}\cdot2^{n-2}$. The following analysis is done for Toffoli gates, with further extensions to Fredkin or other reversible gate families being trivial.
\end{sloppypar}

\begin{sloppypar}
For the ancilla-free reversible logic synthesis, we construct a graph $\mathcal{G} = {\mathcal{V, E}}$, where $\mathcal{V}$ is the set of vertices representing an $n$-variable Boolean 
function as a permutation function over the truth value set $\{0, 1, \ldots 2^{n-1}\}$. An edge from the set of edges ($\mathcal{V}$) represent the application of a $Tof_n$ gate. Since 
the gates are reversible, the edges are not directed. Note that, the set of all permutation functions $\mathcal{S(P)}$ holds \textit{closure} under the application of $Tof_n$ gate. This 
implies that the $\mathcal{G}$ contains cycles. An algorithm for constructing $\mathcal{G}$ is provided as a pseudo-code in the following algorithms~\ref{alg:construct_graph} 
and~\ref{alg:connect_child}.
\end{sloppypar}

\begin{algorithm}
\SetKwInOut{Input}{input}
\SetKwInOut{Output}{output}
\caption{$Create\_Graph$}
\Input{$\mathcal{S(P)}$}
\Output{$\mathcal{G}(\mathcal{V, E})$}
\While{$\mathcal{S(P)}$ $\neq$ $\emptyset$}
{
  $v_{curr} = Create\_Vertex(p \in S(P))$;\\
  $Connect\_Child(v_{curr}$);\\
  \Return $\mathcal{G}(\mathcal{V, E})$;
}
\label{alg:construct_graph}
\end{algorithm}
\decmargin{1em}

\begin{algorithm}
\SetKwInOut{Input}{input}
\SetKwInOut{Output}{output}
\caption{$Connect\_Child$}
\Input{$v_{curr}$}
 $\mathcal{S(P)}$ $\rightarrow$ $remove(perm(v_{curr}))$;\\
 $\mathcal{G} \rightarrow append(v_{curr})$;\\
\lForEach{$gate$ $\in$ $Tof_n$}{
 $v_{res}$ $=$ $Tof_n$($v_{curr}$);\\
 $Connect(v_{res}, v_{curr})$;\\
}
\If{($v_{res}$ $\notin$ $\mathcal{G}$)}{
  $\mathcal{S(P)}$ $\rightarrow$ $remove(v_{res})$;\\
  $Connect\_Child(v_{res})$;\\
}
\label{alg:connect_child}
\end{algorithm}
\decmargin{1em}

\begin{sloppypar}
The function $Create\_Graph$ iterates over all possible permutations $\mathcal{S(P)}$. For each permutation $p$, a vertex is created. For that vertex, all possible $Tof_n$ gate 
is applied. The resulting new permutation is stored in a vertex ($v_{res}$) and connected via an edge to the current vertex. If $v_{res}$ is not already included in $\mathcal{G}$, 
then the function $Connect\_Child$ is called recursively.
\end{sloppypar}

\begin{sloppypar}
The total number of nodes in $\mathcal{G}(\mathcal{V, E})$ is same as the total number of permutation functions for an $n$-variable Boolean function, which is ${2^n}!$. From the 
aforementioned constructed graph, determining the optimal reversible circuit for a given permutation $p_r$ function is nothing but the determination of shortest path from the 
permutation $p_r$ to the input permutation $p_i$. The edges can be assigned with a weight of $1$ or the weight of QC corresponding to the transformation for determining the 
reversible circuit with minimum gate count or minimum QC respectively. The complexity of single-source shortest path computation with edges having non-negative weight is 
$\bigO(E + VlogV)$, where $E$ is the edge count and $V$ is the vertex count~\cite{shortest_path_tarjan}. Clearly, the complexity in our case is exponential. Combining previous 
results, the complexity of optimal ancilla-free reversible logic synthesis for Toffoli network is $\bigO((2^n!)\cdot n\cdot2^{n-1} + (2^n!)log(2^n!))$. 
\end{sloppypar}

\begin{sloppypar}
Figure~\ref{fig:rev_spp} shows parts of a $\mathcal{G}(\mathcal{V, E})$ for a $2$-variable Boolean function. The nodes represent the permutation functions and the edges represent 
the application of different reversible gates. For determining the optimal gate count, the edge weight can be disregarded (every edge weight is $1$) and a simple breadth-first 
search with complexity $\bigO(E)$ can be performed, where $E$ = $(2^n!)\cdot n\cdot2^{n-1}$. For determining minimum QC, however, the algorithm from~\cite{shortest_path_tarjan} 
needs to be applied, which is of higher complexity. 
\end{sloppypar}

\begin{figure}[hbt]
\centering
    \includegraphics[width=80mm]{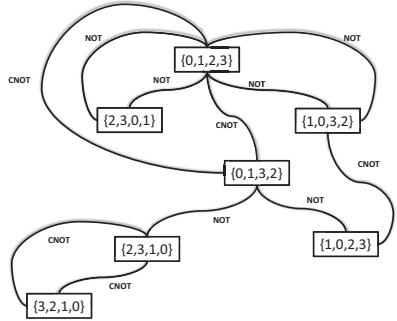}
    \caption{Reversible Logic Synthesis via Shortest Path}
    \label{fig:rev_spp}
\end{figure}

\section{(0-1)-ILP Formulation} \label{sec:ilp_formulation}
\begin{sloppypar}
In this section we present another formulation of ancilla-free reversible logic synthesis as an optimization problem using ILP. An Integer Linear Program (ILP) takes the following form
\begin{equation}
\begin{split}
minimize\: c^Tx \: \\ 
subject\:  to\:  Ax >= b
\end{split}
\end{equation}
, where $A$ is an integer $m\times n$ matrix, $c$ is an integer $n$-vector, $b$ is an integer $m$-vector. The values of $x$, an integer $n$-vector to be determined. When the solution 
set is restricted to either $0$ or $1$, then it becomes a (0-1) ILP, which is known to be NP-complete~\cite{karp_np_complete}.
The (0-1) ILP formulation for reversible logic synthesis has several interconnected modules for expressing the constraints, which are explained in the following.
\end{sloppypar}

\begin{sloppypar}
\textit{Initialization Module}: Corresponding to $n$ variables, $n$ target binary vectors of $2^n$ length are derived. Each target vector ($t_v$) is assigned 
a truth-table expansion corresponding to an input variable.
\end{sloppypar}

\begin{sloppypar}
\textit{Target Selection Module:} All the target binary vectors are interchanged to derive new set of target binary vectors. This permutation network is ensured to preserve the 
exclusivity of the targets. The last line is selected as the target for the following operation. In case of $Fred_n$ gate, the last two lines are selected as the target.
\end{sloppypar}

\begin{sloppypar}
\textit{$Tof_n$ Module:} All but the last target lines are operated on with a conditional AND, where the condition is specified in a vector $c_{tof}$. Based on the result of the 
AND operation, new target binary vectors are derived. Control lines are transported unchanged to new target binary vectors as well. Note that, if the control condition is $0$, the 
AND contributor for that particular control line should be $1$ in order not to disturb the final control value. This reduces to the following specification for a 
control contribution, where $d$ indicates the current depth.
\begin{equation}
\bigwedge_{i=1}^{n-1}(t_v[d][i] == 1)\: \vee \: (c_{tof}[d][i] == 0) 
\end{equation}
This can be interpreted by the optimizer as a NOT gate, when the $c_{tof}$ contains $0$ in each element. To keep the gate count for such assignment and drive the optimizer, 
an additional element $c_{en}$ is introduced as following.
\begin{equation}
c_{en}[d] \: \wedge \: \bigwedge_{i=1}^{n-1}(t_v[d][i] == 1)\: \vee \: (c_{tof}[d][i] == 0) 
\end{equation}
The operational modules are connected via a target selection module. This allows for a compact formulation. The modules are repeated $D$ times, which is maximum possible depth 
for a given variable count, before assigning the target to the output Boolean specification to be synthesized. 
\end{sloppypar}

\begin{sloppypar}
\textit{$Fred_n$ Module:} Similar formulation applies to $Fred_n$ gates as well, where the control values are taken from all but last two target vectors. The condition vectors 
are maintained using $c_{fred}$, which is enabled with an element $c_{en}$ for the particular depth.
\end{sloppypar}

\begin{sloppypar}
The optimization goal can be specified as minimization of gate counts or QC, both of which are based on the condition vector $c_{tof}$. For QC minimization, the $c_{tof}$ condition 
vector is added to identify the number of control lines and corresponding QC value is applied. For gate count 
minimization, the condition vector is checked to contain a single truth-value or not. Correspondingly, a gate is inferred.
\end{sloppypar}

\section{Experimental Results} \label{sec:results}
\begin{sloppypar}
Both the shortest path approach and ILP formulation explained are implemented as standalone C++ program. For the ILP formulation, the API calls for the commercial 
ILP solver IBM CPLEX version 12.2~\cite{ibm_cplex} are utilized. For both the techniques, experiments are done using Intel(R) Core(TM) i7-2670QM CPU @ 2.20GHz machine with 4 GB RAM, 
running Red Hat Linux version 2.6.32-358.6.2.el6.x86\_64.
\end{sloppypar}

\begin{table}[h]
\scriptsize
    \begin{center}
        \caption{Execution Time for Selected Benchmarks}
        \label{tab:ilp_sat}
        \begin{tabular}{c|c|c|c|c}
\hline 
\multirow{2}{*}{Function} & Variable & Gate  & \multicolumn{2}{|c}{Execution Time (seconds)} \\ \cline{4-5}
            & Count    & Count & SMT~\cite{wille_smt} & (0-1) ILP \\ \hline
Peres      	& 3	& 2	&  $<$ 0.01	& 2.16 \\ 
Fredkin		& 3	& 3	&  $<$ 0.01	& 7.19 \\ 
Miller		& 3	& 5	&  $<$ 0.01	& 1536.68 \\
\hline 
        \end{tabular}
    \end{center}
\end{table}

\begin{sloppypar}
For benchmarking the efficiency of the ILP solver against state-of-the-art exact approaches, we compare the execution times with those reported in~\cite{wille_smt}. In~\cite{wille_smt}, 
a Satisfiability Modulo Theory (SMT)-based solver with function-specific optimizations have been used. The results are presented in table~\ref{tab:ilp_sat}. Clearly, (0-1) ILP does 
not provide a practical solution to reversible logic synthesis, despite its capability to explore the optimization space rather than determining a satisfiable solution as in SAT. 
Even for the most basic examples, the runtime quickly increased. To control the runtime, the solution space is restricted in CPLEX by setting {\em IntSolLim} to 1. Even then, it 
produced the result with gap of several orders of magnitude execution time, when compared to SAT.
\end{sloppypar}

\begin{sloppypar}
In contrast, the shortest path problem formulation led to comprehensive results for complete $3$-variable Boolean functions. The flexibility of our solution allowed including 
different gates (N - NOT, C - CNOT, T - Toffoli, F - Fredkin, P - Peres) and also consider positive control and inverted control lines, indicated by '+' and '-' in 
Table~\ref{table:GC_optimal}. Table~\ref{table:QC_optimal} reports corresponding QC values. It can be observed that Fredkin gate is decomposed into Toffoli gates and therefore, 
the QC values are same for NCT and NCTF family, while reporting optimized QC. For each of these gate families, optimal gate count and optimal QC values are obtained within a 
minute for the entire set of $3$-variable functions. Note that, this study only presents the optimal gate/Quantum cost for classical gates. The detailed QC values with 
non-classical gates, as reported in~\cite{rahman_opt_qc}, is avoided because of space limitations.
\end{sloppypar}

\begin{table}[hbt]
\scriptsize
  \caption{Optimal Gate Count results for 3-variable circuits}
  \label{table:GC_optimal}
  \centering
  \begin{tabular}{|c|c|c|c|c|c|c|c|r|}
    \hline
    $GC$ & \multicolumn{2}{c|}{NCT} & \multicolumn{2}{c|}{NCF} & \multicolumn{2}{c|}{NCTF} & \multicolumn{2}{|c|}{NCTPF}\\ \cline{2-9}
             & + & +/- & + & +/- & + & +/- & + & +/-\\
    \hline
    0 & 1  & 1 & 1 & 1 & 1 & 1 & 1 & 1 \\
    1 & 12 & 27 & 12 & 21 & 15 & 33 & 27 & 81\\
    2 & 102 & 369 & 101 & 280 & 143 & 544 & 413 & 2620\\
    3 & 625 & 2925 & 676 & 2422 & 1006 & 4980 & 3910 & 26437 \\
    4 & 2780 & 13282 & 3413 & 11229 & 5021 & 20223 & 17678 & 11181\\
    5 & 8921&  20480 & 11378 & 18689 & 15083 & 14175 & 18073 & 0 \\
    6 & 17049 & 3236 & 17970 & 7558 & 17261 & 364 & 218 & 0 \\
    7 & 10253 & 0 & 6739 & 120 & 1790 & 0 & 0 & 0 \\
    8 & 577 & 0 & 30 & 0 & 0 & 0 & 0 & 0 \\
    \hline
    Avg. & 5.87 & 4.58 & 5.66 & 4.77 & 5.33 & 4.22 & 4.34 & 3.21 \\
    \hline
  \end{tabular}
\end{table}

\begin{table}[hbt]
\scriptsize
  \caption{Optimal QC results for 3-variable circuits}
  \label{table:QC_optimal}
  \centering
  \begin{tabular}{|c|c|c|c|c|c|c|c|r|}
    \hline
    $QC$ & \multicolumn{2}{c|}{NCT} & \multicolumn{2}{c|}{NCF} & \multicolumn{2}{c|}{NCTF} & \multicolumn{2}{|c|}{NCTPF}\\ \cline{2-9}
             & + & +/- & + & +/- & + & +/- & + & +/-\\
    \hline
    0 & 1 & 1 & 1 & 1 & 1 & 1 & 1 & 1\\
    1 & 9 & 15 & 9 & 15 & 9 & 15 & 9 & 15\\
    2 & 51 & 117 & 51 & 117 & 51 & 117 & 51 & 117\\
    3 & 187 & 433 & 187 & 433 & 187 & 443 & 187 & 433\\
    4 & 393 & 534 & 393 & 534 & 393 & 534 & 405 & 582\\
    5 & 477 & 240 & 474 & 228 & 477 & 240 & 609 & 744\\
    6 & 260 & 196 & 215 & 16 & 260  & 196 & 998 & 2572\\
    7 & 338 & 1320 & 17 & 6 & 338 & 1320 & 2648 & 5112\\
    8 & 1335 & 4596 & 48 & 150 & 1335 & 4596 & 4397 & 2400\\
    9 & 3224 & 3168 & 408 & 1479 & 3224 & 3168 & 2712 & 9072\\
    10 & 3686 & 204 & 1919 & 4760 & 3686 & 204 & 5994 & 8456\\
    11 & 902 & 960 & 3931 & 2893 & 902 & 960 & 10249 & 64\\
    12 & 933 & 5712 & 2634 & 120 & 933 & 5712 & 1750 & 7832\\
    13 & 4053 & 10704 & 462 & 0 & 4053 & 10704 & 3488 & 2920\\
    14 & 8690 & 1368 & 5 & 10 & 8690 & 1368 & 6640 & 0\\
    15 & 4903 & 200 & 78 & 468 & 4903 & 200 & 182 & 0\\
    16 & 244  & 2208 & 1038 & 5478 & 244 & 2208 & 0 & 0\\
    17 & 1094 & 6320 & 6079 & 10986 & 1094 & 6320 & 0 & 0\\
    18 & 4346 & 2024 & 9571 & 1874 & 4346 & 2024 & 0 & 0\\
    19 & 4724 & 0 & 2036 & 0 & 4724 & 0 & 0 & 0\\
    20 & 470 & 0 & 12 & 0 & 470 & 0 & 0 & 0\\
    21 & 0 & 0 & 0 & 0 & 0 & 0 & 0 & 0\\
    22 & 0 & 0 & 24 & 288 & 0 & 0 & 0 & 0\\
    23 & 0 & 0 & 732 & 4740 & 0 & 0 & 0 & 0\\
    24 & 0 & 0 & 5496 & 5604 & 0 & 0 & 0 & 0\\
    25 & 0 & 0 & 4482 & 120 & 0 & 0 & 0 & 0\\
    26 & 0 & 0 & 18 & 0 & 0 & 0 & 0 & 0\\
    \hline
    Avg. & 13.74 & 12.46 & 17.47 & 16.58 & 13.74 & 12.46 & 10.52 & 9.35 \\
    \hline
  \end{tabular}
\end{table}

\begin{sloppypar}
The shortest path formulation for reversible logic synthesis provides the computational viewpoint of what has been a purely memory-oriented approach in~\cite{shende_3var} 
and~\cite{golubitsky_4var}. The sheer number of nodes for $4$-variable and beyond presents a bottleneck for scaling the shortest path approach. On the other hand, this may 
provide opportunities for utilizing graph theoretic results such as approximate shortest paths.
\end{sloppypar}

\section{Conclusion and Future Work} \label{sec:conclusion}
\begin{sloppypar}
In this paper, we studied the complexity of exact reversible logic synthesis, for both the cases of the circuit being bounded-ancilla and ancilla-free. For the bounded-ancilla 
scenario, the problem is shown to be NP-hard. For ancilla-free reversible logic synthesis, we linked the problem to sorting by cost-constrained transpositions, for which 
the computational complexity is not known.
\end{sloppypar}

\begin{sloppypar}
We attempted two different formulations for solving the exact problem namely, (0-1) ILP and shortest path problem. The former approach fails to produce optimal solutions in 
acceptable time. The shortest path problem formulation is used to generate comprehensive results for $3$-variable functions, while scaling it up to larger variable sizes remain an 
interesting future work.
\end{sloppypar}


\end{document}